\documentclass[12pt]{article}

\usepackage{amsmath,amsfonts,amsthm} 
\usepackage{braket}
\usepackage{mathrsfs} 
\usepackage[titletoc,toc,title]{appendix}
\usepackage{authblk}
\usepackage{hyperref} 
\usepackage[numbers]{natbib}

\newtheorem{prop}{Proposition}

\title{On the structure of Bethe vectors}
\author[1,2]{J. Fuksa \thanks{\quad fuksajan@fjfi.cvut.cz,  fuksa@theor.jinr.ru}}
\affil[1]{{\small{Department of Mathematics, FNSPE, CTU in Prague, Czech republic}}}
\affil[2]{{\small{Bogoliubov Laboratory of Theoretical Physics, JINR, Dubna, Russia}}}

\begin{document}

\maketitle

\begin{abstract}
The structure of Bethe vectors for generalised models associated with the XXX- and XXZ-type R-matrix is investigated.  The Bethe vectors in terms of two--component and multi--component models are described. Consequently, their structure in terms of local variables and operators is provided. This, as a consequence, proves the equivalence of coordinate and algebraic Bethe ansatzes for the Heisenberg XXX and XXZ spin chains.
\end{abstract}

\section{Introduction}

The quantum inverse scattering method (QISM) was formulated by Faddeev, Sklya\-nin, and Takhtadjan in \citep{TF79,FST80}.  Many physically interesting models were solved by this method such as the one-dimensional Bose gas, the Heisenberg spin chains, the sine-Gordon model \cite{TF79,FST80,Skl79,SF78,IK86}, etc.

To fix the notation the basic features of the QISM are described in section \ref{sec:QISM}. In particular, we recall the notion of the Bethe vectors which are eigenvectors of a family of mutually commuting operators including Hamiltonian of the solved system.  We distinguish between the Bethe vectors, which are eigenvectors of the Hamiltonian and thus their spectral parameters are the solution of the Bethe equations, and the formal Bethe vectors, which have the same formal structure but their spectral parameters may or may not solve the Bethe equations. 

The aim of this article is to describe the structure of the formal Bethe vectors.  For this purpose we describe the generalized two--component and multi--component models introduced in \cite{IK84} and \cite{IzKorRe87}, respectively. The results are formulated as three propositions contained in section \ref{sec:BetheVectors}.

In proposition \ref{prop:2comp} the formal Bethe vectors in terms of the two--component model are constructed. They were obtained in \cite{IK84} for models with the XXX-type R-matrix, including the XXX spin chains and the one--dimensional Bose gas. After minor modifications these vectors are also valid for the models with the XXZ-type R-matrix, i.e., also for the XXZ spin chains, sine--Gordon model, and other models. For its proof see \cite{Slav07}.

Proposition \ref{prop:Kcomp} describes the formal Bethe vectors in terms of the multi--component model. The form of these vectors was firstly published in \cite{IzKorRe87} (see also \cite{KBI93}) for models with the XXX-type R-matrix. It is valid for models associated with the XXZ-type R-matrix as well. We formulate proposition \ref{prop:Kcomp} for the generalised models associated with the R-matrix of both the XXX- and XXZ-type and provide its proof because, as we believe, it is missing in the literature.

Proposition \ref{prop:DetailBV} describes the local structure of the formal Bethe vectors for the generalised inhomogeneous models with the R-matrices of both the XXX- and XXZ-type. We give also its proof. A particular form of proposition \ref{prop:DetailBV} for homogeneous XXX spin chain was obtained in \cite{IzKorRe87}; the proof for its inhomogeneous version can be found in \cite{EFGKK05}.

\section{Quantum inverse scattering method} 
\label{sec:QISM}
\setcounter{equation}{0}

We adopt the notation often used in the context of the QISM. Let $A$ be an operator acting in a vector space $V$, and let $B$ be an operator acting in a tensor product $V\otimes V$. Let us have a tensor product $V^{\otimes N}=V\otimes V\otimes \dotsm \otimes V$ of $N$ spaces $V$. We use the notation $A_i$ for the operator acting as $A$ in the $i$-th copy of $V$ and trivially in the rest of the tensor product, and $B_{kl}$ for the operator acting as $B$ in the tensor product of the $k$-th and $l$-th vector space and trivially in the rest of the tensor product.

We briefly introduce in this section the basic notions of the the QISM. For more details see \cite{Fad96} or \cite{KBI93} and references therein. 

Let us suppose that there is a chain of $N$ sites, each site being endowed with a local Hilbert space $h_j$ and special parameter $\xi_j$. Such a chain is called inhomogeneous. If all parameters $\xi_j$ are equal, it is called homogeneous. The total Hilbert space is $\mathscr{H}=h_1\otimes h_2\otimes\dots\otimes h_N$. To the $j$-th site there corresponds an L-operator $L_j(\mu,\xi_j)=L_j(\mu-\xi_j)$. The L-operator is  a matrix acting in an auxiliary space $V$ and the local Hilbert space $h_j$ of the $j$-th site (quantum space); it depends on two spectral parameters $\mu$ and $\xi_j$.  Depending on a physical model the auxiliary and quantum space vary. The explicit form of the L-operator for particular models can be found in the literature, e.g., \cite{Fad96,KBI93}. The monodromy matrix $T(\lambda)$ is a product of the L-operators along the chain and acts, therefore, in $V\otimes \mathscr{H}$. We restrict ourselves to particular models with an R-matrix of the form \eqref{Rmatrix} below. For such models, the monodromy matrix is of the form
\begin{equation} \label{monodromy}
T(\lambda) = \prod_{j=1}^N L_j(\lambda,\xi_j) = \prod_{j=1}^N \begin{pmatrix}
(L_j(\lambda,\xi_j))_{11} & (L_j(\lambda,\xi_j))_{12} \\
(L_j(\lambda,\xi_j))_{21} & (L_j(\lambda,\xi_j))_{22}
\end{pmatrix} = \begin{pmatrix}
A(\lambda) & B(\lambda) \\
C(\lambda) & D(\lambda)
\end{pmatrix}.
\end{equation}
Its matrix elements $A(\lambda), B(\lambda),C(\lambda), D(\lambda)$ are the operators on $\mathscr{H}$.  Although we do not indicate it explicitly in our notation, the monodromy matrix $T(\lambda)$ as well as its matrix elements depends on the inhomogeneity parameters $\xi_1,\dots,\xi_N$.

The monodromy matrix $T(\lambda)$ satisfies the following bilinear equation (called RTT-relation) in $V_1\otimes V_2 \otimes \mathscr{H}$:
\begin{equation} \label{RTT}
R_{12}(\lambda,\mu) T_1(\lambda) T_2 (\mu) = T_2 (\mu) T_1(\lambda) R_{12}(\lambda,\mu).
\end{equation}
The matrix $R_{12}(\lambda,\mu)$ is called the R-matrix. We suppose here that is of the form 
\begin{equation} \label{Rmatrix}
R(\lambda,\mu)=\begin{pmatrix}
f(\lambda,\mu) & 0 & 0 & 0 \\
0 & 1 & g(\lambda,\mu) & 0 \\
0 & g(\lambda,\mu) & 1 & 0 \\
0 & 0 & 0 & f(\lambda,\mu)
\end{pmatrix}.
\end{equation}
There are two R-matrices of the form \eqref{Rmatrix}. The first one is rational with the matrix elements  
\begin{equation} \label{Rxxx}
f(\lambda,\mu)= \frac{\lambda-\mu+1}{\lambda-\mu}, \qquad g(\lambda,\mu)= \frac{1}{\lambda-\mu}
\end{equation}
and corresponds, e.g., to the XXX spin chain and the one-dimensional Bose gas. The second one is trigonometric with the matrix elements
\begin{equation} \label{Rxxz}
f(\lambda,\mu)= \frac{\sinh(\lambda-\mu+\eta)}{\sinh(\lambda-\mu)}, \qquad g(\lambda,\mu)= \frac{\sinh \eta}{\sinh(\lambda-\mu)}
\end{equation}
and corresponds, e.g., to the XXZ spin chain and the sine--Gordon model.

The R-matrix \eqref{Rmatrix} satisfies the famous Yang--Baxter equation in the tensor product of three auxiliary spaces $V_1\otimes V_2 \otimes V_3$
\begin{equation}\label{YBequation}
R_{12}(\lambda_1 , \lambda_2) R_{13}(\lambda_1,\lambda_3) R_{23}(\lambda_2,\lambda_3) = R_{23}(\lambda_2,\lambda_3) R_{13}(\lambda_1,\lambda_3) R_{12}(\lambda_1,\lambda_2).
\end{equation}

Let us mention that the L-operator $L(\lambda,\xi)$ satisfies the same relation \eqref{RTT}, as can be easily seen if we restrict ourselves in \eqref{monodromy} to the chain of the length one. In fact, relation \eqref{RTT} for the monodromy matrix is a consequence of the same relation for the L-operators, see \cite{Fad96}.

Relation \eqref{RTT} determines an algebra with bilinear commutation relations. The R-matrix \eqref{Rmatrix} is an analogue of its structure constants and the Yang--Baxter equation \eqref{YBequation} is an analogue of the Jacobi identity. Let us list some of the relations here:
\begin{align}
[T_{jk}(\lambda),T_{jk}(\mu)] & = 0,\qquad j,k=1,2, \label{TTcomm}\\
A(\mu) B(\lambda) & = f(\lambda,\mu) B(\lambda) A(\mu) + g(\mu,\lambda) B(\mu) A(\lambda),\label{ABcomm}\\
B(\mu) A(\lambda) & = f(\lambda,\mu) A(\lambda) B(\mu) + g(\mu,\lambda) A(\mu) B(\lambda), \label{BAcomm}\\
D(\mu) B(\lambda) & = f(\mu,\lambda) B(\lambda) D(\mu) + g(\lambda,\mu) B(\mu) D(\lambda), \label{DBcomm}\\
& \text{and another seven relations} \label{Restcomm}
\end{align}
which can be found, e.g., in review papers \cite{Fad96,Slav07}. 

Relation \eqref{RTT} is the starting point of the quantum inverse scattering method. What  follows is a general procedure which is related not only to the R-matrix of the form \eqref{Rmatrix}.

The R-matrix \eqref{Rmatrix} is invertible almost everywhere. From this fact one can  immediately prove the commutation relation
\begin{equation} \label{CommutationTransfer}
[\mathscr{T}(\lambda),\mathscr{T}(\mu)]=0
\end{equation}
for the transfer matrix
\begin{equation} \label{TransferMatrix}
\mathscr{T}(\lambda) \equiv \mathrm{Tr} (T(\lambda)) =A(\lambda) + D(\lambda).
\end{equation}
Indeed, we can write \eqref{RTT} as $R_{12}(\lambda,\mu) T_1(\lambda) T_2 (\mu) R^{-1}_{12}(\lambda,\mu) = T_2 (\mu) T_1(\lambda) $ and taking traces in the spaces $V_1$ and $V_2$ we get  \eqref{CommutationTransfer}.

The transfer matrix \eqref{TransferMatrix} constitutes a generating function for a class of $N-1$ commuting operators. The QISM diagonalizes all these commuting operators simultaneously, as it diagonalizes their generating function $\mathscr{T}(\lambda)$. The eigenvectors resulting from this diagonalisation are called the Bethe vectors.  We mention here only the result. For details see \cite{Fad96} (see also \cite{BFIKN15}). For diagonalisation of $\mathscr{T}(\lambda)$, it is necessary to suppose that the Hilbert space $\mathscr{H}$  has the structure of the Fock space with the cyclic vector $\ket{0}$ called pseudovacuum. Let the pseudovacuum $\ket{0}$ be an eigenstate of both the operators $A(\lambda)$ and $D(\lambda)$ and be annihilated by the operator $C(\lambda)$:
\begin{equation} \label{assumptions1}
A(\lambda) \ket{0}= a(\lambda)\ket{0},\qquad D(\lambda)\ket{0}=d(\lambda)\ket{0},\qquad C(\lambda)\ket{0} = 0. 
\end{equation} 
Such a model with unspecified functions $a(\lambda)$ and $d(\lambda)$ is called generalized.

The eigenvectors of the transfer matrix $\mathscr{T}(\mu)$ are of the form
\begin{equation} \label{BetheVec}
\ket{\{\lambda\}} = \prod_{j=1}^M B(\lambda_j) \ket{0},
\end{equation}
if the parameters $\lambda_k$ are pairwise distinct: $\lambda_j\neq \lambda_k$ for $j\neq k$ and satisfy the Bethe equations
\begin{equation} \label{BetheEq}
\mathscr{Y}(\lambda_k|\{\lambda\})=0
\end{equation} 
for all $k=1,\dots,M$. Here, $M$ is a number of excitations, $0\leq M\leq N$. The corresponding eigenvalue is
\begin{equation}
\tau(\mu|\{\lambda\}) = a(\mu) \prod_{a=1}^M f(\lambda_a,\mu) + d(\mu) \prod_{a=1}^M f(\mu,\lambda_a)
\end{equation}
and the function $\mathscr{Y}(\mu|\{\lambda\})$ appearing in the Bethe equation is of the form:
\begin{equation}
\mathscr{Y}(\mu|\{\lambda\}) = \tau(\mu|\{\lambda\}) \prod_{a=1}^M g^{-1}(\lambda_a,\mu)
\end{equation}
where the functions $f(\lambda,\mu)$ and $g(\lambda,\mu)$ are the matrix elements of the R-matrix \eqref{Rmatrix}. The vectors \eqref{BetheVec} are called the Bethe vectors. 

We distinguish here the Bethe vectors and formal Bethe vectors: the former are eigenvectors of the transfer matrix, whereas the latter have the structure of \eqref{BetheVec} but the spectral parameters of their creation-like operators $B(\lambda_j)$ are arbitrary, i.e., they may or may not be eigenvectors of the transfer matrix $\mathscr{T}(\mu)$. We will deal with the formal Bethe vectors for the rest of the text.

\section{Structure of formal Bethe vectors} 
\label{sec:BetheVectors}
\setcounter{equation}{0}

The aim of this section is to describe the detailed structure of the formal Bethe vectors, i.e., of vectors of the form \eqref{BetheVec}. For models associated with the XXX-type R-matrix, a lot was known from the origins of the QISM \cite{IK84,IzKorRe87}. It helped with the identification of the QISM with the former coordinate Bethe ansatz method \cite{Be31} for the homogeneous XXX spin chains. We intend to generalize known results and to provide the structure of Bethe vectors for generalised inhomogeneous models associated with the R-matrix of the XXX- and the XXZ-type. We suppose, as in section \ref{sec:QISM}, that the length of the chain is $N$.

\subsection{Multi--component model and formal Bethe vectors}
\label{ssec:ManyComp}

The two--component model was developed for calculations of correlation functions of a local operator sitting at the site  $x$. The chain is split into two subchains of the length $x$ and $N-x$ and we define correspondingly the monodromy matrix for each subchain. For details see \citep{IK84}. This can be obviously generalized to an arbitrary number $K\leq N$ of subchains \cite{IzKorRe87}:
\begin{equation}
T(\lambda) = T(\lambda|1)T(\lambda|2)\cdots T(\lambda|K).
\end{equation}
The total Hilbert space $\mathscr{H}$ is divided into its $K$ subspaces $\mathscr{H}=\mathscr{H}_1\otimes \mathscr{H}_2\otimes \cdots\otimes \mathscr{H}_K$. Consequently, the pseudovacuum is split into $\ket{0}=\ket{0}_1\otimes \ket{0}_2\otimes\cdots\otimes \ket{0}_K$, $\ket{0}_j\in \mathscr{H}_j$, where $j=1,2,\dots,K$. The partial monodromy matrix
\begin{equation}
T(\lambda|j)=\begin{pmatrix}
A_j(\lambda) & B_j(\lambda) \\
C_j(\lambda) & D_j(\lambda)
\end{pmatrix},
\end{equation}
$j=1,2,\dots,K$, satisfies the same RTT-condition \eqref{RTT} as the undivided monodromy matrix $T(\lambda)$ with the same R-matrix \eqref{Rmatrix}. Its matrix elements act non-trivially only on the Hilbert subspace $\mathscr{H}_j$ and due to \eqref{RTT} satisfy the same set of bilinear relations \eqref{TTcomm}--\eqref{Restcomm} as the matrix elements of the full model. The matrix elements of different partial monodromy matrices $T(\lambda|k)$ and $T(\lambda|j)$ mutually commute. 

We suppose that the operators $A_j(\lambda),D_j(\lambda), C_j(\lambda)$ for $j=1,\dots,K$ act on partial pseudovacuum $\ket{0}_j \in \mathscr{H}_j$ as
\begin{align} \label{PartOperators}
 A_j(\lambda)\ket{0}_j =a_j(\lambda) \ket{0}_j, \quad D_j(\lambda)\ket{0}_j=d_j(\lambda) \ket{0}_j,\quad  C_j(\lambda)\ket{0}_j=0.
\end{align} 
The functions $a_j(\lambda),d_j(\lambda)$ are arbitrary; therefore, we are dealing with the generalised model. The eigenfunctions of $A(\lambda)$ and $D(\lambda)$ are then products of the eigenfunctions of the components
\begin{equation} \label{PartEigenval}
a(\lambda) = \prod_{j=1}^K a_j(\lambda), \quad d(\lambda) = \prod_{j=1}^K d_j(\lambda).
\end{equation}

The authors of \cite{IK84,IzKorRe87}, state that the formal Bethe vectors \eqref{BetheVec} of the full model can be expressed in terms of the formal Bethe vectors of its components for models associated with the XXX-type R-matrix. After minor modifications this is also valid for models with the XXZ-type R-matrix as well. We formulate below two propositions about the formal Bethe vectors for the two--component and the multi--component generalised models. 

\begin{prop} \label{prop:2comp}
Let us have a system with $K=2$ components. Let $\ket{\{\lambda \}}$ be the formal Bethe vector \eqref{BetheVec} of the system. It can be expressed in terms of two components as:
\begin{align}\label{2comp}
	& \ket{\{\lambda \}} = 
	 \sum_{\mathcal{J}\in\mathcal{P}(M)} \prod_{k_1\in \mathcal{J}} \prod_{k_2\in \bar{\mathcal{J}}} f(\lambda_{k_1},\lambda_{k_2}) \, d_2(\lambda_{k_1}) \, a_1(\lambda_{k_2}) \; B_1(\lambda_{k_1}) B_2(\lambda_{k_2})  \ket{0}
\end{align}
where $f(\lambda,\mu)$ is the matrix element of the R-matrix \eqref{Rmatrix} and the summation is performed over all sets $\mathcal{J}$ from the power set $\mathcal{P}(M)$ of the set $\{1,\dots,M\}$. $\bar{\mathcal{J}}$ is the complement of $\mathcal{J}$ in $\{1,\dots,M\}$:  $\mathcal{J} \cup  \bar{\mathcal{J}} =\{1,\dots,M\}$, $\mathcal{J} \cap  \bar{\mathcal{J}} =\emptyset$.
\end{prop}
For proof see \cite{Slav07}. The result of proposition \ref{prop:2comp} can be generalized for an arbitrary number of components $K\leq N$, as stated in \cite{IzKorRe87}. 

\begin{prop}\label{prop:Kcomp} Let us have a system of $K\leq N$ components. An arbitrary formal Bethe vector \eqref{BetheVec} of the full system can be expressed in terms of the formal Bethe vectors of its $K$ components as:
\begin{align}
	& \ket{\{\lambda \}} = 
 \sum_{\mathcal{P}_K(M)}  \prod_{k_1\in \mathcal{J}_1} \prod_{k_2\in \mathcal{J}_2} \dotsm \prod_{k_K\in \mathcal{J}_K} \prod_{1\leq i<j\leq K} \Bigl( a_i(\lambda_{k_j}) d_j(\lambda_{k_i}) f(\lambda_{k_i},\lambda_{k_j}) \Bigr) \nonumber\\
  & \qquad \times B_1(\lambda_{k_1}) B_2(\lambda_{k_2})  \dotsb B_K(\lambda_{k_K}) \ket{0}\label{Kcomp}
\end{align}
where the summation is performed over the set $\mathcal{P}_K(M)$ of all divisions of the set $\{1,\dots,M\}$ into its $K$  subsets $\mathcal{J}_1,\mathcal{J}_2,\dots, \mathcal{J}_{K}$ such that: $\mathcal{J}_1\cup \mathcal{J}_2\cup \dots \cup \mathcal{J}_{K}=\{1,\dots,M\}$ and $\mathcal{J}_j\cap \mathcal{J}_k=\emptyset$ for $j\neq k$.
\end{prop}
The proof is missing in the literature. We give it here.
\begin{proof}
We perform the proof by induction on the number of components $K$. For $K=2$, eq. \eqref{Kcomp}  coincides with \eqref{2comp}. 

Let us suppose that \eqref{Kcomp} is valid for $K-1<N$. The chain of the length $N$ is divided into $K-1$ subchains. Consequently, the Hilbert space is divided into $K-1$ subspaces $\mathscr{H}=\mathscr{H}_1\otimes\cdots \otimes \mathscr{H}_{K-1}$. 

Let us mention that we do not specify the explicit form of the division into $K-1$ subchains. We suppose that the induction hypothesis holds for all possible divisions into $K-1$ subchains. It is obvious that all possible divisions into $K$ subchains are obtained by dividing the last subchain  of the divisions into $K-1$ subchains, if possible, into its two subchains. 

To the division of the $(K-1)$-st subchain into its two subchains there corresponds a division of the Hilbert space $\mathscr{H}_{K-1}=\mathscr{H}'_{K-1}\otimes \mathscr{H}'_K$  into its two subspaces $\mathscr{H}'_{K-1}$ and $\mathscr{H}'_{K}$. Consequently, the  pseudovacuum $\ket{0}_{K-1}=\ket{0}'_{K-1}\otimes \ket{0}'_{K}$ is divided into two pseudovacua $\ket{0}'_{K-1}$ and $\ket{0}'_{K}$. The igenvalues of $A_{K-1}(\lambda)$, resp., $D_{K-1}(\lambda)$  are divided into parts as in \eqref{PartEigenval}: $a_{K-1}(\lambda)=a'_{K-1}(\lambda) a'_K(\lambda)$, resp., $d_{K-1}(\lambda)=d'_{K-1}(\lambda) d'_K(\lambda)$. Using this and proposition \ref{prop:2comp}, we get
\begin{align}
	& \prod_{k_{K-1}\in \mathcal{J}_{K-1}} B_{K-1}(\lambda_{k_{K-1}}) \ket{0}_{K-1} \notag\\
	& =   \sum \prod_{j_{K-1}\in \mathcal{J}'_{K-1}} \prod_{j_{K}\in \mathcal{J}'_{K}} f(\lambda_{j_{K-1}},\lambda_{j_K}) d'_{K}(\lambda_{j_{K-1}}) a'_{K-1}(\lambda_{j_K})  \nonumber\\
	&\qquad \times B'_{K-1}(\lambda_{j_{K-1}})  \ket{0}'_{K-1}\otimes B'_{K}(\lambda_{j_K}) \ket{0}'_{K} \label{Ncomp1}
\end{align}
where the sum goes over all divisions of $\mathcal{J}_{K-1}$ into its two disjoint subsets $\mathcal{J}'_{K-1}$ and $\mathcal{J}_{K}'$ such that $\mathcal{J}'_{K-1} \cup \mathcal{J}_{K}'=\mathcal{J}_{K-1}$. The operators $B'_{K-1}(\lambda)$ and $B'_{K}(\lambda)$ act on new subchains with Hilbert spaces $\mathscr{H}'_{K-1}$ and $\mathscr{H}'_K$, respectively. Using this in the induction hypothesis for $K-1$ we prove  \eqref{Kcomp}. 
\end{proof}

\subsection{Local structure of formal Bethe vectors}
\label{ssec:LocalBV}

We investigate in this subsection the local structure of the formal Bethe vectors. The approach  developed in this subsection is the application of  proposition \ref{prop:Kcomp} for subchains of the length one (1-subchains). We divide the chain of the length $N$ into its $N$ 1-subchains. The local Hilbert space corresponding to the $j$-th 1-subchains is $h_j$. The monodromy matrix $T_j(\lambda)$  is  identical with the L-operator $L_j(\lambda, \xi_j)$, as can be seen from \eqref{monodromy}. 

Let the local quantum space $h_j$ contain a pseudovacuum $\ket{0}_j$. Let $\ket{0}_j$  be a common eigenvector of the diagonal elements of $L_j(\lambda,\xi_j)$ and be annihilated by the subdiagonal element: 
\begin{align}
& (L_j(\lambda,\xi_j))_{11} \ket{0}_j = \alpha(\lambda,\xi_j) \ket{0}_j, \quad (L_j(\lambda,\xi_j))_{22} \ket{0}_j = \delta(\lambda,\xi_j) \ket{0}_j, \notag \\ 
& (L_j(\lambda,\xi_j))_{21} \ket{0}_j = 0. \label{assumptions2}
\end{align}
Assumptions \eqref{assumptions2} for the L-operators $L_j(\lambda,\xi_j)$ ensure the validity of assumptions \eqref{PartOperators} of multi--component models for arbitrary $K$, $K\leq N$. Consequently, the validity of propositions \ref{prop:2comp} and \ref{prop:Kcomp}. Particularly, \eqref{assumptions2} ensures the validity of assumptions \eqref{assumptions1} necessary for the diagonalisation procedure. 

The complete pseudovacuum $\ket{0}\in \mathscr{H}$ is of the  form of an $N$-fold tensor product 
\begin{equation} \label{pseudovacuum}
\ket{0}= \ket{0}_1  \otimes \ket{0}_2 \otimes \cdots \otimes \ket{0}_N
\end{equation}
and the eigenvalues of $A(\lambda)$ and $D(\lambda)$ in \eqref{assumptions1} are:
\begin{equation}
a(\lambda) = \prod_{j=1}^N \alpha(\lambda,\xi_j), \quad d(\lambda) = \prod_{j=1}^N \delta(\lambda,\xi_j).
\end{equation}
The creation-like operator $B_j(\lambda)$ of the $j$-th 1-subchain is identical with the overdiagonal element of $L_j(\lambda,\xi)$
\begin{equation} \label{B1-subchain}
B_j(\lambda) = (L_j(\lambda,\xi_j))_{12}.
\end{equation}
We restrict our consideration to a very specific representation in which 
\begin{equation} \label{assumption3}
B_j(\lambda) B_j(\mu) \ket{0} = 0.
\end{equation}
This condition is satisfied for both the XXX and XXZ spin chains, see \cite{Fad96} (see also \cite{BFIKN15}), and indicates the fermionic type behavior of the operators $B_j(\lambda)$.

\begin{prop}\label{prop:DetailBV}
Let assumptions \eqref{assumptions2} be satisfied. Let creation-like operators \eqref{B1-subchain} satisfy condition \eqref{assumption3}. Then formal Bethe vector \eqref{BetheVec} can be represented in the form
\begin{align}
  \ket{\{\lambda\}} & = \sum_{1\leq n_1 < n_2 <\dots< n_M\leq N} \, \sum_{\sigma_\lambda \in S_M} \sigma_\lambda \Biggl(  B_{n_1}(\lambda_{1}) B_{n_2}(\lambda_{2}) \dotsb B_{n_M}(\lambda_{M})\ket{0}   \nonumber\\
	& \times \prod_{l=1}^M  \prod_{i=1}^{n_l-1}\alpha(\lambda_l,\xi_i) \prod_{j=n_l+1}^N \delta(\lambda_l,\xi_j) \prod_{r=l+1}^{M} f(\lambda_l,\lambda_r)  \Biggr) \label{Bethe1}
\end{align}
where $S_M$ is the permutation group and its elements $\sigma_\lambda$ permute the spectral parameters $\{\lambda_1,\dots,\lambda_M\}$.
\end{prop}

\begin{proof}
We realised above that assumptions \eqref{assumptions2} ensure the validity of proposition \ref{prop:Kcomp} for $K=N$.  Condition \eqref{assumption3} provides that the sets $\mathcal{J}_k$, $k=1,\dots,N$, from proposition \ref{prop:Kcomp} for $K=N$ are of the cardinality maximally one: $|\mathcal{J}_k|\leq 1$. Moreover, only $M$ of these sets are nonempty, let us denote them $\mathcal{J}_{n_1},\mathcal{J}_{n_2},\dots,\mathcal{J}_{n_M}$. The summation over all distributions $\mathcal{P}_N(M)$ in \eqref{Bethe1} is then equivalent to the summation over all $n_1,n_2,\dots,n_M$ such that $1\leq n_1 < n_2 < \cdots < n_M \leq N$, and over all permutations $\sigma_\lambda\in S_M$ of the set of spectral-parameters $\{\lambda_1,\dots,\lambda_M\}$  provided that $k\in \mathcal{J}_{n_k}$. This observation drastically simplifies the subsequent considerations.

Let us study what happens with the coefficient
\begin{equation}
\mathscr{C}_a = \prod_{k_1\in\mathcal{J}_1} \cdots \prod_{k_N\in\mathcal{J}_N } \prod_{1\leq i<j\leq N}  a_i(\lambda_{k_j})
\end{equation}
appearing in \eqref{Kcomp} for $K=N$. We know with respect to the above considerations that only $\mathcal{J}_{n_1},\mathcal{J}_{n_2},\dots,\mathcal{J}_{n_M}$ from $\mathcal{J}_{1},\mathcal{J}_{2},\dots,\mathcal{J}_{N}$ are nonempty and, moreover, contain only one element. Therefore,
\begin{equation}
\mathscr{C}_a = \prod_{l=1}^M \prod_{i=1}^{n_l-1} a_i(\lambda_{k_{n_l}}).
\end{equation}
Moreover, we have supposed that $k_{n_l}=l$ and $a_i(\lambda)=\alpha(\lambda,\xi_i)$. Hence,
\begin{equation}
\mathscr{C}_a = \prod_{l=1}^M \prod_{i=1}^{n_l-1} \alpha(\lambda_l,\xi_i).
\end{equation}
Similarly, we obtain that
\begin{align}
\mathscr{C}_d & = \prod_{k_1\in\mathcal{J}_1} \cdots \prod_{k_N\in\mathcal{J}_N } \prod_{1\leq i<j\leq N}  d_j(\lambda_{k_i}) = \prod_{l=1}^M \prod_{j=n_l+1}^{N} \delta(\lambda_l,\xi_j), \\
\mathscr{C}_f & = \prod_{k_1\in\mathcal{J}_1} \cdots \prod_{k_N\in\mathcal{J}_N } \prod_{1\leq i<j\leq N} f(\lambda_{k_i},\lambda_{k_j}) = \prod_{l=1}^M \prod_{r=l+1}^{M} f(\lambda_l,\lambda_r).
\end{align}
The product of the partial Bethe vectors is 
\begin{align}
\prod_{k_1\in\mathcal{J}_1} \cdots \prod_{k_N\in\mathcal{J}_N }  B_1(\lambda_{k_1})  B_2(\lambda_{k_2}) \cdots B_N(\lambda_{k_N}) = 
B_{n_1}(\lambda_{1})  B_{n_2}(\lambda_{2})  \cdots B_{n_M}(\lambda_M).
\end{align}
Gluing these results together we prove \eqref{Bethe1}.
\end{proof}

For models, where $B_j(\lambda)$ are parameter independent, expression \eqref{Bethe1} can be further simplified. This holds particularly for the XXX and XXZ spin chains, cf. \cite{Fad96} or \cite{BFIKN15}.  For homogeneous spin chains we obtain the following representation of the formal Bethe vectors:
\begin{align} \label{BetheVecDetail}
\ket{\{\lambda\}} & = \prod_{l=1}^M \frac{\delta^N(\lambda_l,\xi)}{\alpha(\lambda_l,\xi)}  \sum_{1\leq n_1<\cdots<n_M\leq N}  B_{n_1}B_{n_2}\cdots B_{n_M} \ket{0} \\
&  \times \sum_{\sigma_\lambda \in S_M} \sigma_\lambda  \left(  \prod_{1\leq i<j\leq M} f(\lambda_{i},\lambda_{j}) \prod_{k=1}^M \left( \frac{\alpha(\lambda_k,\xi)}{ \delta(\lambda_k,\xi)} \right)^{n_k}  \right). \notag
\end{align} 
This holds for the homogeneous XXX and XXZ spin chains.

The representation of the formal Bethe vectors \eqref{BetheVecDetail} was published in \cite{IzKorRe87} for the homogeneous XXX spin chain. The special form of proposition \ref{prop:DetailBV} for the inhomogeneous XXX spin chain is given in \cite{EFGKK05}. Proposition \ref{prop:DetailBV} provides the generalisation of the known results for the generalised inhomogeneous models associated with the R-matrix of both the XXX- and XXZ-type.

\section{Conclusions}

The results of section \ref{sec:BetheVectors} are valid for the generalised inhomogeneous models associated with the R-matrix of the XXX- and XXZ-type. 

The representation \eqref{BetheVecDetail} proves the equivalence of the QISM to the coordinate Bethe ansatz for the homogeneous XXX and XXZ spin chains, as was noted  in \cite{IzKorRe87} for the XXX spin chain (see also \cite{BFIKN15}).

Formula \eqref{Bethe1} is the most general formula giving an explicit expression for the formal Bethe vectors of all models satisfying condition \eqref{assumption3}, like the inhomogeneous XXX spin chain, the inhomogeneous XXZ chain, or the fermionic realisation of the homogeneous  XXX and XXZ spin chains (see \cite{BFIKN15,BFI14} and references therein).

\section*{Acknowledgement}

The author would like to express his gratitude to  his scientific leaders A. Isaev and \v{C}. Burd\'{i}k for their guidance in his research and to N. Slavnov for the  advice to study two--component models in the context of the Bethe vectors. 

The work of the author was supported by the grant of the CTU in Prague $SGS15/215/OHK4/3T/14$ and by the Grant of the Plenipotentiary of the Czech Republic at JINR, Dubna.


\phantomsection

\addcontentsline{toc}{section}{References}

\bibliographystyle{abbrv}
\bibliography{literatura}

\begin{thebibliography}{10}

\bibitem{Be31}
H.~Bethe.
\newblock Zur theorie der metalle.
\newblock {\em Zeitschrift f{\"u}r Physik}, 71(3):205--226, 1931.

\bibitem{BFI14}
{\v{C}. Burd\'{i}k, J. Fuksa and A. Isaev}.
\newblock Bethe vectors for xxx-spin chain.
\newblock {\em Journal of Physics: Conference Series}, 563(1):012011, 2014.

\bibitem{BFIKN15}
{\v{C}. Burd\'{i}k, J. Fuksa, A. P. Isaev, S. O. Krivonos, and O.
  Navr\'{a}til}.
\newblock {Remarks towards the spectrum of the Heisenberg spin chain type
  models}.
\newblock {\em Phys. Part. Nuclei}, 46(3):277--309, 2015.
\newblock arXiv:1412.3999v2.

\bibitem{EFGKK05}
F.~Essler, H.~Frahm, F.~G{\"o}hmann, A.~Kl{\"u}mper, and V.~E. Korepin.
\newblock {\em {The one-dimensional Hubbard model}}.
\newblock Cambridge University Press, 2005.

\bibitem{Fad96}
L.~D. Faddeev.
\newblock {How algebraic Bethe ansatz works for integrable model}.
\newblock In {\em Symétries quantiques (Les Houches)}, pages 149--219, 1996.

\bibitem{FST80}
L.~D. Faddeev, E.~K. Sklyanin, and L.~A. Takhtajan.
\newblock {The quantum inverse problem method. 1}.
\newblock {\em Theoret. and Math. Phys.}, 40:688, 1980.

\bibitem{IK86}
A.~Izergin and V.~Korepin.
\newblock {The Pauli principle for one-dimensional bosons and the algebraic
  Bethe ansatz}.
\newblock {\em Journal of Soviet Mathematics}, 34(5):1933--1937, 1986.

\bibitem{IK84}
A.~G. Izergin and V.~E. Korepin.
\newblock {The quantum inverse scattering method approach to correlation
  functions}.
\newblock {\em Comm. Math. Phys.}, 94(1):67--92, 1984.

\bibitem{IzKorRe87}
A.~G. Izergin, V.~E. Korepin, and N.~Y. Reshetikhin.
\newblock Correlation functions in a one-dimensional bose gas.
\newblock {\em Journal of Physics A: Mathematical and General},
  20(14):4799--4822, 1987.

\bibitem{KBI93}
V.~E. Korepin, N.~M. Bogoliubov, and A.~G. Izergin.
\newblock {\em {Quantum inverse scattering mathod and correlation functions}}.
\newblock Cambridge Univ. Press., 1993.

\bibitem{Skl79}
E.~K. Sklyanin.
\newblock {Method of the inverse scattering problem and the nonlinear quantum
  Schr{\"o}dinger equation}.
\newblock {\em Sov. Phys. Dokl.}, 24:107, 1979.

\bibitem{SF78}
E.~K. Sklyanin and L.~D. Faddeev.
\newblock {Quantum mechanical approach to completely integrable field theory
  models}.
\newblock {\em Sov. Phys. Dokl}, 23(1):978, 1978.

\bibitem{Slav07}
N.~A. Slavnov.
\newblock {The algebraic Bethe ansatz and quantum integrable systems}.
\newblock {\em Russ. Math. Surv.}, 62(4):727, 2007.

\bibitem{TF79}
L.~A. Takhtajan and L.~D. Faddeev.
\newblock {The quantum method of the inverse problem and the Heisenberg XYZ
  model}.
\newblock {\em Russ. Math. Surv.}, 34(5):11, 1979.

\end{thebibliography}

\end{document}